\documentclass[10pt]{article}
\usepackage{epsf}
\usepackage{amsmath}

\allowdisplaybreaks

\usepackage[showframe=false]{geometry}
\usepackage{changepage}

\usepackage{epsfig}
\usepackage{amssymb}

\usepackage{amsthm}
\usepackage{setspace}
\usepackage{cite}
\usepackage{mcite}

\usepackage{algorithmic}  
\usepackage{algorithm}

\usepackage{shadow}
\usepackage{fancybox}
\usepackage{fancyhdr}

\usepackage{color}
\usepackage[usenames,dvipsnames,svgnames,table]{xcolor}

\definecolor{mypurple}{rgb}{.4,.0,.5}

\usepackage{xcolor}

\usepackage{color}

\definecolor{darkgreen}{rgb}{0, 0.4,0}
\newcommand{\dgr}[1]{\textcolor{darkgreen}{#1}}

\definecolor{purplebrown}{rgb}{0.5,0.1,0.6}

\newcommand{\bl}[1]{\textcolor{blue}{#1}}
\newcommand{\red}[1]{\textcolor{red}{#1}}
\newcommand{\prp}[1]{\textcolor{purple}{#1}}

\definecolor{shadebrown}{rgb}{0.1,0.1,0.9}
\definecolor{lightblue}{rgb}{0.2,0,1}

\usepackage{fancybox}
\usepackage{graphicx}
\usepackage{epstopdf}
\usepackage{epsfig}
\usepackage{wrapfig}
\usepackage{subfigure}

\usepackage{xcolor}

\usepackage{tcolorbox}
\tcbuselibrary{skins}

\newtcbox{\xmyboxQ}{on line,
arc=7pt,
before upper={\rule[-3pt]{0pt}{10pt}},boxrule=1.1pt,
boxsep=0pt,left=6pt,right=6pt,top=0pt,bottom=0pt,enhanced, frame style image=blueshade.png,interior style image=goldshade.png}

\newtcbox{\xmybox}{on line,
arc=7pt,
before upper={\rule[-3pt]{0pt}{10pt}},boxrule=0pt,
boxsep=0pt,left=6pt,right=6pt,top=0pt,bottom=0pt,enhanced, coltext=blue, colback=white!10!yellow}

\newtcbox{\xmyboxa}{on line,
arc=7pt,
before upper={\rule[-3pt]{0pt}{10pt}},boxrule=0pt,
boxsep=0pt,left=6pt,right=6pt,top=0pt,bottom=0pt,enhanced, colback=white!10!yellow}

\newtcbox{\xmyboxb}{on line,
arc=7pt,
before upper={\rule[-3pt]{0pt}{10pt}},boxrule=1pt,colframe=darkgreen!100!blue,
boxsep=0pt,left=6pt,right=6pt,top=0pt,bottom=0pt,enhanced, colback=white!10!yellow}

\newtcbox{\xmyboxc}{on line,
arc=7pt,
before upper={\rule[-3pt]{0pt}{10pt}},boxrule=.7pt,colframe=blue!100!blue,
boxsep=0pt,left=6pt,right=6pt,top=0pt,bottom=0pt,enhanced, coltext=blue, colback=white!10!yellow}

\newtcbox{\xmytboxa}{on line,
arc=7pt,
before upper={\rule[-3pt]{0pt}{10pt}},boxrule=.0pt,colframe=pink!50!yellow,
boxsep=0pt,left=6pt,right=6pt,top=0pt,bottom=0pt,enhanced, coltext=white, colback=blue!40!red}

\newtcbox{\xmytboxb}{on line,
arc=7pt,
before upper={\rule[-3pt]{0pt}{10pt}},boxrule=.0pt,colframe=pink!50!yellow,
boxsep=0pt,left=6pt,right=6pt,top=0pt,bottom=0pt,enhanced, coltext=white, colback=white!40!green}

\usepackage[hyphens]{url}

\usepackage[colorlinks=true,
            linkcolor=black,
            urlcolor=blue,
            citecolor=purple]{hyperref}

\usepackage{breakurl}

\def\s{{\bf s}}

\def\y{{\bf y}}
\def\v{{\bf v}}
\def\x{{\bf x}}

\def\x{{\mathbf x}}

\def\s{{\bf s}}

\def\v{{\bf v}}
\def\x{{\bf x}}
\def\y{{\bf y}}
\def\z{{\bf z}}

\def\h{{\bf h}}

\def\be{\begin{equation}}
\def\ee{\end{equation}}
\def\ba{\left[\begin{array}}
\def\ea{\end{array}\right]}

\def\s{{\bf s}}

\def\v{{\bf v}}
\def\x{{\bf x}}
\def\y{{\bf y}}
\def\z{{\bf z}}

\def\1{{\bf 1}}

\def\0{{\bf 0}}

\def\erf{\mbox{erf}}
\def\erfc{\mbox{erfc}}

\def\mR{{\mathbb R}}

\def\mE{{\mathbb E}}

\newtheorem{theorem}{Theorem}

\begin{document}

\begin{singlespace}

\title {Starting CLuP with polytope relaxation  
}
\author{
\textsc{Mihailo Stojnic
\footnote{e-mail: {\tt flatoyer@gmail.com}} }}
\date{}
\maketitle

\centerline{{\bf Abstract}} \vspace*{0.1in}

The Controlled Loosening-up (CLuP) mechanism that we recently introduced in \cite{Stojnicclupint19} is a generic concept that can be utilized to solve a large class of problems in polynomial time. Since it relies in its core on an iterative procedure, the key to its excellent performance lies in a typically very small number of iterations needed to execute the entire algorithm. In a separate paper \cite{Stojnicclupcmpl19}, we presented a detailed complexity analysis that indeed confirms the relatively small number of iterations. Since both papers, \cite{Stojnicclupint19} and \cite{Stojnicclupcmpl19} are the introductory papers on the topic we made sure to limit the initial discussion just to the core of the algorithm and consequently focused only on the algorithm's most basic version. On numerous occasions though, we emphasized that various improvements and further upgrades are possible. In this paper we present a first step in this direction and discuss a very simple upgrade that can be introduced on top of the basic CLuP mechanism. It relates to the starting of the CLuP and suggests the well-known so-called polytope-relaxation heuristic (see, e.g. \cite{StojnicBBSD05,StojnicBBSD08}) as the starting point. We refer to this variant of CLuP as the CLuP-plt and proceed with the presentation of its complexity analysis. As in \cite{Stojnicclupcmpl19}, a particular \textbf{\emph{complexity analysis per iteration level}} type of complexity analysis is chosen and presented through the algorithm's application on the well-known MIMO ML detection problem. As expected, the analysis confirms that CLuP-plt performs even better than the original CLuP. In some of the most interesting regimes it often achieves within the \textbf{\emph{first three iterations}} an excellent performance. We also complement the theoretical findings with a solid set of numerical experiments. Those also happen to be in an excellent agreement with the analytical predictions.

\vspace*{0.25in} \noindent {\bf Index Terms: Controlled Loosening-up (CLuP); Polytope relaxation; ML - detection; MIMO systems; Algorthms; Random duality theory}.

\end{singlespace}

\section{Introduction}
\label{sec:back}

To handle famous MIMO ML detection problem, we in \cite{Stojnicclupint19} presented the so-called Controlled Loosening-up (CLuP) algorithm. Since the CLuP algorithm will be the main topic of this paper as well, and since we will study its behavior when applied for solving the MIMO ML detection problems, we first briefly recall on the basics of the MIMO ML.

As usual, one start with the most basic linear system:
\begin{eqnarray}\label{eq:linsys1}
\y=A\x_{sol}+\sigma\v,
\end{eqnarray}
where  $\y\in\mR^m$ is the output vector, $A\in\mR^{m\times n}$ is the system matrix , $\x_{sol}\in\mR^n$ is the input vector, $\v\in\mR^m$ is the noise vector at the output, and $\sigma$ is a scaling factor that determines the ratio of the useful signal and the noise (the so-called SNR (signal-to-noise ratio)). It goes without saying that this type of system modeling is among the most useful/popular in various scientific/engineering fields (a particularly popular application of this model in the fields of information theory and signal processing is its utilization in modeling of multi-antenna systems).

Also, we will here continue the trend that we have started in \cite{Stojnicclupint19} and \cite{Stojnicclupcmpl19}, and consider a statistical setup where both $\v$ and $A$ are comprised of i.i.d. standard normal random variables. A similar continuing the trend from \cite{Stojnicclupint19} and \cite{Stojnicclupcmpl19} regarding the so-called \emph{linear} regime will be in place as well. That means that in this paper we will also view $n$ and $m$ as large but with a constant proportionality between them, i.e. we will assume that $m=\alpha n$ where $\alpha\in\mR_+$ is a number that doesn't change as both $n$ and $m$ grow. The following optimization problem is the simplest yet most fundamental version of the MIMO ML-detection problem
\begin{eqnarray}\label{eq:ml1}
\hat{\x}=\min_{\x\in{\cal X}}\|\y-A\x\|_2,
\end{eqnarray}
where ${\cal X}$ is the set of all available input vectors $\x$. Now, many interesting scenarios/variants of the MIMO ML problem appear depending on the structure of ${\cal X}$ (for example, LASSO/SOCP variants of (\ref{eq:ml1}) often seen in statistics, machine learning, and compressed sensing are just a tiny subset of many very popular scenarios of interest; more on these considerations can be found in e.g. \cite{StojnicGenLasso10,CheDon95,Tibsh96,DonMalMon10,BunTsyWeg07,vandeGeer08,MeinYu09}). Here, we follow into the footsteps of \cite{Stojnicclupint19,Stojnicclupcmpl19} and consider the standard information theory/wireless communications binary scenario which assumes ${\cal X}=\{-\frac{1}{\sqrt{n}},\frac{1}{\sqrt{n}}\}^n$. It goes trivially, basically almost without saying, that $\x_{sol}\in {\cal X}$ is naturally assumed as well. We will also without a loss of generality assume even further that $\x_{sol}=\{\frac{1}{\sqrt{n}},\frac{1}{\sqrt{n}},\dots,\frac{1}{\sqrt{n}}\}$.

The above problem (\ref{eq:ml1}) can be solved either exactly or approximately (for more on various relaxing heuristics see, e.g. \cite{GolVanLoan96Book,GroLovSch93Book,vanMaarWar00,GoeWill95}). What makes it particularly interesting is that in the above mentioned binary scenario, (\ref{eq:ml1}) is typically viewed as a very hard combinatorial optimization type of problem. As such it was obviously the topic of interest in various research communities over last at least half a century. Many excellent algorithms and algorithmic heuristics have been introduced over this period of time. As a detailed discussion about such developments is more suited for review papers we here just in passing mention that some of the very best results regarding to the perspective of the problem that is of interest here can be found in e.g. \cite{StojnicBBSD05,StojnicBBSD08,FinPhoSD85,HassVik05,JalOtt05}. In addition, we also emphasize two probably the most important points: 1) the problem in (\ref{eq:ml1}) is hard if one wants/needs to solve it \textbf{exactly} and 2) polynomial heuristics typically offer an approximate solution that does trail the exact one by a solid margin in almost all interesting scenarios. In \cite{Stojnicclupint19} and \cite{Stojnicclupcmpl19} we introduced the above mentioned CLuP mechanism as a way of attacking the MIMO ML on the exact level. Compared to \cite{StojnicBBSD05,StojnicBBSD08}, which also attacked the MIMO ML on the exact level, CLuP did so by running only a few (fairly often not more than $10$) simplest possible quadratic programming type of iterations. In \cite{Stojnicclupcmpl19}, this rather remarkable property was analytically characterized. Here we provide a similar characterization for a slightly different upgraded variant of CLuP. Before we proceed with the characterization of this new CLuP variant we below first recall on the CLuP's basics.

\subsection{CLuP's basics}
\label{sec:clupbasics}

As is by now well-known from \cite{Stojnicclupint19,Stojnicclupcmpl19}, CLuP is effectively a very simple iterative procedure that in its core form assumes choosing a starting  $\x^{(0)}\in{\cal X}=\{-\frac{1}{\sqrt{n}},\frac{1}{\sqrt{n}}\}^n$, radius $r$, and running the following
\begin{eqnarray}
\x^{(i+1)}=\frac{\x^{(i+1,s)}}{\|\x^{(i+1,s)}\|_2} \quad \mbox{with}\quad \x^{(i+1,s)}=\mbox{arg}\min_{\x} & & -(\x^{(i)})^T\x  \nonumber \\
\mbox{subject to} & & \|\y-A\x\|_2\leq r\nonumber \\
&& \x\in \left [-\frac{1}{\sqrt{n}},\frac{1}{\sqrt{n}}\right ]^n. \label{eq:clup1}
\end{eqnarray}
As one can guess, the choice for $\x^{(0)}$ and $r$ has a very strong effect on the way how the algorithm progresses. For the simplest possible choice of $\x^{(0)}$ (each component of $\x^{(0)}$ is generated with equal likelihood as $-\frac{1}{\sqrt{n}}$ or $\frac{1}{\sqrt{n}}$) we in Figure \ref{fig:fignum2} show both, the theoretical and the simulated CLuP's performance.
\begin{figure}[htb]
\centering
\centerline{\epsfig{figure=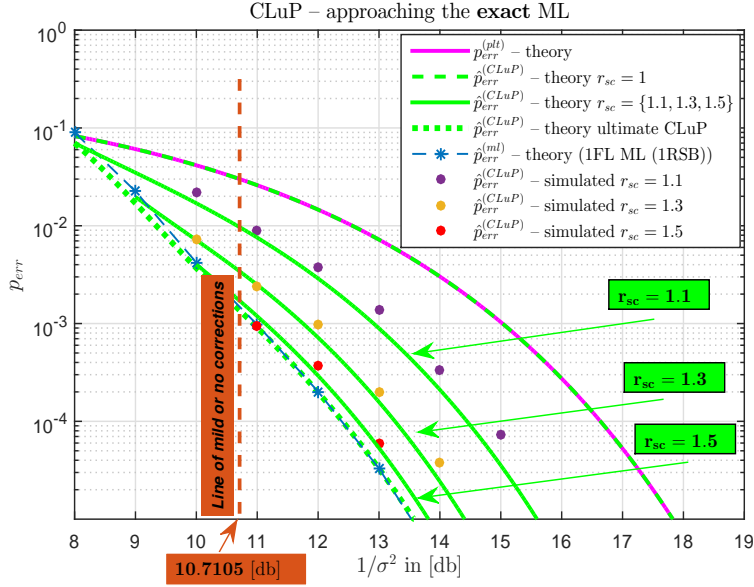,width=11.5cm,height=8cm}}
\caption{$p_{err}$ as a function of $1/\sigma^2$; $\alpha=0.8$ -- theory and simulations}
\label{fig:fignum2}
\end{figure}
Without going into too much details, we just briefly mention that as $r$ increases from $1.1r_{plt}$ to $1.5r_{plt}$ (more on the definition and importance of $r_{plt}$ can be found in \cite{Stojnicclupint19}) CLuP's performance gets closer to the ML. Further detailed explanations related to the figure can be found in \cite{Stojnicclupint19}. Those among other things include a discussion regarding the appearance of a vertical line (the so-called line of corrections). We of course skip repeating such discussion and just mention that in this paper (similarly to \cite{Stojnicclupcmpl19}) we will be interested in the regimes above the line, i.e. in the regimes where the SNR, $1/\sigma^2$, is to the right of the line.

What is of a bit more interest to the present paper though (and what can't exactly be seen from Figure \ref{fig:fignum2}) is the complexity of the above CLuP algorithm. The analysis of the CLuP's complexity was of course the main topic of \cite{Stojnicclupcmpl19}. The remarkable CLuP's property that it fairly often runs not only in a polynomial but rather fixed number of iterations was through such an analysis fully characterized. What may have escaped the attention in \cite{Stojnicclupcmpl19} is the fact that not only is the number of CLuP's iterations fixed and small, it is actually achieved without much effort in making the algorithm even the tiniest of the bits more complex than its most basic version. That in the first place meant that in \cite{Stojnicclupcmpl19}, we analyzed CLuP's complexity by assuming that the starting $\x^{(0)}$ is basically completely random and as such in a way completely disconnected from the problem at hand. On the other hand, it seems rather natural that a bit more clever choice could help CLuP achieve even better performance. There are a tone of possible choices for $\x^{(0)}$ and the next natural question would be which of such choices would be the best or at least better than the random one. Such a discussion requires a careful analysis and we will present it in a separate companion paper. To insure that the initial discussion in this direction is as simple as possible, we here focus on a particular choice of the starting $\x^{(0)}$ that we view as pretty much the simplest, most natural one after the fully random one considered in \cite{Stojnicclupint19,Stojnicclupcmpl19}.

\subsection{CLuP-plt}
\label{sec:secclupplt}

The new CLuP's variant that we consider in this paper (and to which we refer as CLuP-plt), assumes simply generating $\x^{(0)}$ as the solution to the standard polytope-relaxation heuristic (see, e.g. \cite{StojnicBBSD05,StojnicBBSD08}) of the original ML problem (\ref{eq:ml1})
\begin{eqnarray}
\x^{(0,plt)}=\mbox{arg}\min_{\x} & & \|\y-A\x\|_2  \nonumber \\
\mbox{subject to} & & \x\in \left [-\frac{1}{\sqrt{n}},\frac{1}{\sqrt{n}}\right ]^n. \label{eq:clup0}
\end{eqnarray}
Then one can define CLuP-plt as
\begin{eqnarray}
\x^{(i+1)}=\frac{\x^{(i+1,s)}}{\|\x^{(i+1,s)}\|_2} \quad \mbox{with}\quad \x^{(i+1,s)}=\mbox{arg}\min_{\x} & & -(\x^{(i)})^T\x  \nonumber \\
\mbox{subject to} & & \|\y-A\x\|_2\leq r\nonumber \\
&& \x\in \left [-\frac{1}{\sqrt{n}},\frac{1}{\sqrt{n}}\right ]^n, \label{eq:clup1}
\end{eqnarray}
where $i$ starts from zero and $\x^{(0)}=\x^{(0,plt)}$. Alternatively, one can increment the indices and start counting the iterations by first setting
\begin{equation}\label{eq:clup2}
 \x^{(1,s)}= \x^{(0,plt)}
\ \qquad \mbox{and}\qquad \x^{(1)}= \frac{\x^{(1,s)}}{\|\x^{(1,s)}\|_2}= \frac{\x^{(0,plt)}}{\|\x^{(0,plt)}\|_2},
\end{equation}
and then continuing with (\ref{eq:clup1}) for $i\geq 1$. To be in an alignment with what we have done in \cite{Stojnicclupcmpl19} and to accurately account for the (\ref{eq:clup0}) as the first iteration of the algorithm (as we should) we will rely on (\ref{eq:clup2}) and (\ref{eq:clup1}) with $i\geq 1$.

The main idea behind the CLuP-plt introduced above is that the initial $\x^{(0,plt)}$ is expected to be closer to the targeted optimal solution and as such might help getting to the optimum faster. Below we will provide an analysis that will confirm these expectations. We will formally focus on the algorithm's complexity, which due to its iterative nature amounts to handling the number of iterations. However, we will present a particular type of analysis that we typically refer to as the \textbf{\emph{complexity analysis per iteration level}}, where we basically fully characterize all system parameters and how they change through each of the running iterations. Such an analysis is of course way more demanding than just mere computation of the total number of needed iterations.

Through the presentation below we will see that the analysis of CLuP-plt can be designed so that it to a large degree parallels what we have done when we analyzed the complexity of the original CLuP in \cite{Stojnicclupcmpl19}. We will therefore try to avoid repeating many explanations that are to a large degree similar or even identical to the corresponding ones in \cite{Stojnicclupcmpl19} and instead focus on the key differences. Also, we will emphasize it on multiple occasions but do mention it here as well that we chose a very simple upgrade to showcase potential of the CLuP's core mechanism. Since we will be utilizing the main concepts of the analysis from \cite{Stojnicclupcmpl19} in some of our companion papers as well, we also found this particular upgrade as a very convenient choice to quickly get fully familiar with all the key steps of \cite{Stojnicclupcmpl19}. In a way, we will essentially through a reconsideration in this paper bring those steps (that at first may appear complicated) to a level of a routine. This will turn out to be particularly useful when we switch to discussion of a bit more advanced structures.

Parallelling what was done in \cite{Stojnicclupcmpl19} the presentation will be split into several parts. The characterization of the algorithms's first iteration will be briefly discussed at the beginning and then in the second part we will move to the second and higher iterations. We will also present a large set of simulations results and observe that they are in a rather nice agreement with the theoretical findings.

\section{Complexity analysis of CLuP-plt -- first iteration}
\label{sec:clupfirstit}

As mentioned above, to facilitate the exposition and following we will try to parallel as much as possible the flow of the presentation from \cite{Stojnicclupcmpl19}. That means that the core of the complexity analysis will again be the so-called \textbf{\emph{complexity analysis on per iteration level}}.

We start things off by noting that a combination of (\ref{eq:linsys1}) and (\ref{eq:clup0}) gives the following version of the CLuP-plt's first iterations
\begin{eqnarray}
\x^{(0,plt)}=\mbox{arg}\min_{\x} & & \|\sigma\v+A(\x_{sol}-\x)\|_2  \nonumber \\
\mbox{subject to} & & \x\in \left [-\frac{1}{\sqrt{n}},\frac{1}{\sqrt{n}}\right ]^n, \label{eq:avoidclup1}
\end{eqnarray}
which with a cosmetic change easily becomes
\begin{eqnarray}
\x^{(0,plt)}=\mbox{arg}\min_{\x} & & \|\sigma\v+A\z\|_2  \nonumber \\
\mbox{subject to} & & \z\in \left [0,\frac{2}{\sqrt{n}}\right ]^n. \label{eq:avoidclup1plt1}
\end{eqnarray}
Following considerations from \cite{Stojnicclupint19,Stojnicclupcmpl19} and ultimately those
from \cite{StojnicCSetam09,StojnicCSetamBlock09,StojnicISIT2010binary,StojnicDiscPercp13,StojnicUpper10,StojnicGenLasso10,StojnicGenSocp10,StojnicPrDepSocp10,StojnicRegRndDlt10,Stojnicbinary16fin,Stojnicbinary16asym} and utilizing the concentration strategy we set $\|\z\|_2^2=c_{1,z}$ and instead of (\ref{eq:avoidclup1plt1}) consider
\begin{eqnarray}
\xi_{p,1}(\alpha,\sigma,c_{1,z})=\lim_{n\rightarrow\infty}\frac{1}{\sqrt{n}}\mE\min_{\z} & & \|\sigma\v+A\z\|_2  \nonumber \\
\mbox{subject to} & & \|\z\|_2^2=c_{1,z}\nonumber \\
&& \z\in \left [0,2/\sqrt{n}\right ]^n. \label{eq:avoidclup1a3}
\end{eqnarray}
It is now not that hard to note that the problem in (\ref{eq:avoidclup1a3}) is conceptually identical to the one in equation (7) in \cite{Stojnicclupcmpl19}. In fact, it can be thought of a special case of the one from \cite{Stojnicclupcmpl19} with $s_1$ and the components of $\x^{(0)}$ in equation (7) in \cite{Stojnicclupcmpl19} being equal to zero. This basically means that one can completely repeat the rest of the analysis from the second section of \cite{Stojnicclupcmpl19}. The only substantial difference will be that the $\nu$ variable from \cite{Stojnicclupcmpl19}'s second section will now be zero. In particular, instead of \cite{Stojnicclupcmpl19}'s equation (16) one now has for the optimizing $\z_i$
\begin{equation}
\z_i=\frac{1}{\sqrt{n}}\min \left (\max\left (0,-\left (\frac{\h}{2\gamma}\right )\right ),2\right ).\label{eq:avoidclup14a}
\end{equation}
Moreover, analogously to \cite{Stojnicclupcmpl19}'s equations (18) and (19) one now has
\begin{eqnarray}
I_{1,1}(\gamma) &  = & -(exp(-0.5(4 \gamma )^2) ( - 4 \gamma) + \sqrt{\pi/2}\erf(2\sqrt{2}\gamma + 1/\sqrt{2}))/(4 \sqrt{2 \pi}\gamma) \nonumber \\
I_{2,1}(\gamma)  &  = & 2\gamma\erfc((4\gamma)/\sqrt{2})-2exp(-1/2(4\gamma)^2)/\sqrt{2\pi},
\label{eq:avoidclup16}
\end{eqnarray}
and
\begin{equation}
\xi_{RD}^{(1)}(\alpha,\sigma;c_{1,z},\gamma)=\sqrt{\alpha}\sqrt{c_{1,z}+\sigma^2}+I_{1,1}(\gamma)+I_{2,1}(\gamma)-\gamma c_{1,z}. \label{eq:avoidclup17}
\end{equation}
The following theorem summarizes what we presented above.
\begin{theorem}(CLuP-plt -- RDT estimate -- first iteration)
Let $\xi_{p,1}(\alpha,\sigma,c_{1,z})$ and $\xi_{RD}^{(1)}(\alpha,\sigma;c_{1,z},\gamma)$ be as in (\ref{eq:avoidclup1a3}) and (\ref{eq:avoidclup17}), respectively. Then \begin{equation}
\xi_{p,1}(\alpha,\sigma,c_{1,z})= \max_{\gamma}\xi_{RD}^{(1)}(\alpha,\sigma;c_{1,z},\gamma).\label{eq:thmavoidcluprd1}
\end{equation}
Consequently,
\begin{equation}
\min_{c_{1,z}}\xi_{p,1}(\alpha,\sigma,c_{1,z})= \min_{c_{1,z}}\max_{\gamma}\xi_{RD}^{(1)}(\alpha,\sigma;c_{1,z},\gamma).
\label{eq:thmcluprd2}
\end{equation}\label{thm:avoidcluprd1}
\end{theorem}
\begin{proof}
Follows automatically from \cite{Stojnicclupcmpl19} and ultimately the RDT mechanisms from \cite{StojnicCSetam09,StojnicISIT2010binary,StojnicDiscPercp13,StojnicGenLasso10,StojnicGenSocp10,StojnicPrDepSocp10,StojnicRegRndDlt10} (as in \cite{Stojnicclupcmpl19}, the \bl{\textbf{strong random duality}} is trivially in place here as well).
\end{proof}
We do mention in passing also that one can trivially first solve the optimization over $c_{1,z}$ and effectively transform/simplify the above optimization problem to an optimization over only $\gamma$. However, to maintain parallelism with \cite{Stojnicclupcmpl19} and ultimately with what we will present below, we avoided doing so.

\subsection{CLuP-plt -- first iteration summary}
\label{sec:clupfirstitsummary}

Since the above theorem is very similar to the corresponding one in \cite{Stojnicclupcmpl19}, we below continue to follow into the footsteps of \cite{Stojnicclupcmpl19} and in a summarized way formalize how  it can be utilized to finally obtain all of the key algorithm's parameters in the first iteration.

\vspace{.1in}
\noindent \xmyboxc{\bl{\emph{\textbf{Summary of the CLuP-plt's first iteration }}}}

We first solve
\begin{eqnarray}
\{\hat{\gamma}^{(1)},\hat{c}_{1,z}^{(1)}\}=\mbox{arg} \min_{0\leq c_{1,z}\leq 4}\max_{\gamma}\xi_{RD}^{(1)}(\alpha,\sigma;c_{1,z},\gamma)\label{eq:avoidclup17a}
\end{eqnarray}
and then as in \cite{Stojnicclupcmpl19}'s equations (23) define
\begin{eqnarray}
s_{x,1}(\gamma) & = & 1/2/\gamma/\sqrt{2\pi}(1-exp(-(4\gamma)^2/2))\nonumber \\
s_{xsq,1}(\gamma) & = & -I_{1,1}(\gamma)/\gamma\nonumber \\
s_{x,2}(\gamma) & = & 2(.5\erfc((4\gamma)/\sqrt{2}))\nonumber \\
s_{xsq,2}(\gamma) & = & 2s_{x,2}.\label{eq:avoidclup18}
\end{eqnarray}
Moreover, analogously to \cite{Stojnicclupcmpl19}'s (24),(25), and (26) we now have
\begin{eqnarray}
\sqrt{n}\mE\z_i & = &  s_{x,1}(\hat{\gamma}^{(1)})+ s_{x,2}(\hat{\gamma}^{(1)})\nonumber \\
n\mE\z_i^2 & = & s_{xsq,1}(\hat{\gamma}^{(1)})+ s_{xsq,2}(\hat{\gamma}^{(1)}),\label{eq:avoidclup19}
\end{eqnarray}
and with $\x_i=\x_{sol}-\z_i$ also
\begin{eqnarray}
\sqrt{n}\mE\x_i & = & 1-(s_{x,1}(\hat{\gamma}^{(1)})+ s_{x,2}(\hat{\gamma}^{(1)}))\nonumber \\
n\mE\x_i^2 & = & s_{xsq,1}(\hat{\gamma}^{(1)})+ s_{xsq,2}(\hat{\gamma}^{(1)})+\sqrt{n}2\mE\x_i-1,\label{eq:avoidclup20}
\end{eqnarray}
and finally
\begin{eqnarray}
\mE((\x_{sol})^T\x) & = & 1-(s_{x,1}(\hat{\gamma}^{(1)})+ s_{x,2}(\hat{\gamma}^{(1)}))\nonumber \\
\mE\|\x\|_2^2  & = & s_{xsq,1}(\hat{\gamma}^{(1)})+ s_{xsq,2}(\hat{\gamma}^{(1)})+2\mE((\x_{sol})^T\x)-1.\label{eq:avoidclup21}
\end{eqnarray}
As in \cite{Stojnicclupcmpl19}, the strong random duality ensures that the above are not only the expected values but also the concentrating points of the corresponding quantities (concentration of course is exponential in $n$). As in \cite{Stojnicclupcmpl19}'s (27) one can also obtain for the probability of error
\begin{equation}\label{eq:avoidclup21a}
  p_{err}^{(1)}=1-P\left (\z_i\leq \frac{1}{\sqrt{n}}\right )=1-\frac{1}{2}\erfc\left ( \frac{-2\hat{\gamma}^{(1)}}{\sqrt{2}}\right ).
\end{equation}

The theoretical values for all key system parameters that can be obtained utilizing the above Theorem \ref{thm:avoidcluprd1} are shown
in Table \ref{tab:tabavoidnum1} for SNR, $1/\sigma^2=13$[db].
\begin{table}[h]
\caption{\textbf{Theoretical} values for key system parameters obtained based on Theorem \ref{thm:avoidcluprd1}}\vspace{.1in}
\hspace{-0in}\centering
\small{
\begin{tabular}{||c||c|c||c|c||c|c|c|c||}\hline\hline
$1/\sigma^2 $[db]  & $\hat{\nu}^{(1)}$ & $\hat{\gamma}^{(1)}$ & $\hat{c}_{1,z}^{(1)}$ & $\hat{s}_1^{(1)}$ &   $\xi_{RD}^{(1)} $ & $p_{err}^{(1)} $  & $\|\x^{(1,s)}\|_2^2$ &
$(\x_{sol})^T\x^{(1,s)}$ \\ \hline\hline
$13  $ & $\mathbf{0}  $ & $\mathbf{1.2233}  $ & $\mathbf{0.0835}  $ & $\mathbf{-0}  $ & $\mathbf{0.1226}  $ & $\mathbf{0.0072}  $ & $\mathbf{0.7574}  $ & $\mathbf{0.8369}  $ \\ \hline\hline
\end{tabular}}
\label{tab:tabavoidnum1}
\end{table}
To maintain the parallelism with \cite{Stojnicclupcmpl19} and with what we will present below, we artificially keep two additional parameters  $\hat{\nu}^{(1)}$ and $\hat{s}_1^{(1)}$ and assign the value $0$ to them.

\section{Summary of the CLuP-plt's second iteration analysis}
\label{sec:clupsecondit}

The move from the first to the second iteration is of course of critical importance for understanding all later moves from $k$-th to $(k+1)$-th iteration for $k\geq 2$.
The CLuP's second iteration assumes computing $\x^{(2)}$ as
\begin{eqnarray}
\x^{(2)}=\frac{\x^{(2,s)}}{\|\x^{(2,s)}\|_2} \quad \mbox{with}\quad \x^{(2,s)}=\mbox{arg}\min_{\x} & & -(\x^{(1)})^T\x  \nonumber \\
\mbox{subject to} & & \|\y-A\x\|\leq r\nonumber \\
&& \x\in \left [-\frac{1}{\sqrt{n}},\frac{1}{\sqrt{n}}\right ]^n, \label{eq:avoidsecclup1}
\end{eqnarray}
where we recall from (\ref{eq:clup2}), $\x^{(1)}= \frac{\x^{(1,s)}}{\|\x^{(1,s)}\|_2}= \frac{\x^{(0,plt)}}{\|\x^{(0,plt)}\|_2}$. One can then also rewrite (\ref{eq:avoidsecclup1}) in the following way
\begin{eqnarray}
\min_{\z} & & (\x^{(1)})^T\z  \nonumber \\
\mbox{subject to} & & \|\sigma\v+A\z\|\leq r\nonumber \\
&& \z\in \left [0,2/\sqrt{n}\right ]^n. \label{eq:avoidsecclup1a2}
\end{eqnarray}
Utilizing once again the concentration strategy we set $\|\z\|_2^2=c_{2,z}$ and $(\x^{(1,s)})^T\z=s_2$ and consider
\begin{eqnarray}
\xi_{p,2}(\alpha,\sigma,c_{2,z},s_2)=\lim_{n\rightarrow\infty}\frac{1}{\sqrt{n}}\mE\min_{\z} & & \|\sigma\v+A\z\|_2  \nonumber \\
\mbox{subject to} & & \|\z\|_2^2=c_{2,z}\nonumber \\
& & (\x^{(1,s)})^T\z=s_2 \nonumber \\
&& \z\in \left [0,2/\sqrt{n}\right ]^n. \label{eq:avoidsecclup1a3}
\end{eqnarray}
The above problem is structurally literally identical to \cite{Stojnicclupcmpl19}'s (31). One can then repeat all the steps between \cite{Stojnicclupcmpl19}'s (31) and (56) to arrive at the following set of equations that determine the optimizing $\z_i$ and $\x_i^{(2,s)}$
\begin{eqnarray}
\z_i^{(2)} & = & \frac{1}{\sqrt{n}}\min \left (\max\left (0,-\left (\frac{\h_i^{(1,p)}+\nu\x^{(1,s)}_i+\nu_2}{2\gamma}\right )\right ),2\right )\nonumber \\
\x_i^{(2,s)} & = & \frac{1}{\sqrt{n}}-\z_i^{(2)}=\frac{1}{\sqrt{n}}\left (1-\min \left (\max\left (0,-\left (\frac{\h_i^{(1,p)}+\nu\x^{(1,s)}_i+\nu_2}{2\gamma}\right )\right ),2\right )\right ),\label{eq:avoidsecclup14a}
\end{eqnarray}
where one also recalls from (\ref{eq:clup2})
\begin{equation}
\x_i^{(1,s)}=\x^{(0,plt)}=1-\z_i^{(1)}=\frac{1}{\sqrt{n}}\left (1-\min \left (\max\left (0,-\left (\frac{\h_i}{2\hat{\gamma}^{(1)}}\right )\right ),2\right )\right ).\label{eq:avoidsecclup15a}
\end{equation}
As in \cite{Stojnicclupcmpl19}, $\h^{(1,p)}=p^{(1)}\h+\sqrt{1-(p^{(1)})^2}\h^{(1)}$ and the components of both $\h$ and $\h^{(1)}$ are i.i.d. standard normals. Setting
\begin{equation}
I_{1}^{(2)}(\gamma,\nu,\nu_2)  =  \int\int((\h_i^{(1,p)}+\nu\x^{(1,s)}+\nu_2)\z_i^{(2)}+\gamma\left (\z_i^{(2)}\right )^2)exp\left (-\frac{\left (\h_i^{(1)}\right )^2+ \h_i^2}{2}\right )\frac{d\h_i^{(1)}d\h_i}{2\pi},
\label{eq:avoidsecclup16}
\end{equation}
where if negative, the term under the integral is zero for $\gamma<0$. Analogously to \cite{Stojnicclupcmpl19}'s (60) one can define
\begin{eqnarray}
\xi_{RD}^{(2)}(\alpha,\sigma;p^{(1)},q^{(1)},c_{2,z},s_2,s_3,\gamma,\nu,\nu_2) & = & \sqrt{\alpha}\sqrt{c_{2,z}+\sigma^2}\left (q^{(1)}p^{(1)}+\sqrt{1-(q^{(1)})^2}\sqrt{1-(p^{(1)})^2}\right )\nonumber \\
& & +I_{1}^{(2)}(\gamma,\nu,\nu_2) -\nu s_2-\nu_2 s_3-\gamma c_{2,z}. \label{eq:avoidsecclup17}
\end{eqnarray}
Finally, from \cite{Stojnicclupcmpl19}'s (76)-(78) one has
\begin{eqnarray}
\phi_b^{(2)}=\mbox{arg} \min_{s,d_1^{(2)},d_2^{(2)}} & & s\nonumber \\
\mbox{subject to} & & \max_{p^{(1)}}\min_{0\leq c_{2,z}\leq 4}\max_{\gamma,\nu,\nu_2}\xi_{RD}^{(2)}(\alpha,\sigma;p^{(1)},q^{(1)},c_{2,z},s_2,s_3,\gamma,\nu,\nu_2)=r\nonumber \\
& & s_2=d_1^{(1)}+s\sqrt{d_2^{(1)}}\nonumber \\
& & s_3=1-d_1^{(2)}\nonumber \\
& & c_{2,z}=d_2^{(2)}-2d_1^{(2)}+1\nonumber \\
& & q^{(1)} =\frac{s_3-s_2+\sigma^2}{\sqrt{c_{2,z}+\sigma^2}\sqrt{\hat{c}_{1,z}+\sigma^2}},\label{eq:sumsecavoidclup17a5}
\end{eqnarray}
and
\begin{equation}\label{eq:sumsecavoidclup17a7}
 p_{err}^{(2)}=1-\int\int((\mbox{sign}(\x^{(2,s)})+1)/2)exp\left (-\frac{\left (\h_i^{(1)}\right )^2+ \h_i^2}{2}\right )\frac{d\h_i^{(1)}d\h_i}{2\pi}.
\end{equation}
To obtain the remaining key parameters one can utilize
 \begin{eqnarray}\label{eq:sumsecavoidclup17a8}
 \hat{d}_{2}^{(2)} & = & \int\int((\x_i^{(2,s)})^2exp\left (-\frac{\left (\h_i^{(1)}\right )^2+ \h_i^2}{2}\right )\frac{d\h_i^{(1)}d\h_i}{2\pi}\nonumber \\
  \hat{d}_{1}^{(2)} & = & \int\int((\x_i^{(2,s)})exp\left (-\frac{\left (\h_i^{(1)}\right )^2+ \h_i^2}{2}\right )\frac{d\h_i^{(1)}d\h_i}{2\pi}\nonumber \\
    \hat{s}_{2}^{(2)} & = & \int\int((\x_i^{(1,s)})\z_i^{(2)}exp\left (-\frac{\left (\h_i^{(1)}\right )^2+ \h_i^2}{2}\right )\frac{d\h_i^{(1)}d\h_i}{2\pi}.
\end{eqnarray}
To make things easier to follow one can define a set of the key output parameters at the end of the second iteration (of course, it goes without emphasizing that $\x^{(2,s)}$ is the main output of the second iteration). This set consists of \bl{critical} plus \prp{auxiliary} parameters
\begin{equation}\label{eq:sumavoidclup21aa}
  \phi^{(2)}=\{\bl{p_{err}^{(2)},\hat{s}^{(2)},\hat{d}_2^{(2)},\hat{d}_1^{(2)}},\prp{\hat{\nu}^{(2)},\hat{\nu}_2^{(2)},\hat{\gamma}^{(2)},\hat{p}^{(1)},\hat{q}^{(1)},\hat{c}_{2,z}^{(1)},\hat{s}_{2}^{(2)},\hat{s}_{3}^{(2)}}\},
\end{equation}
where
\begin{eqnarray}\label{eq:sumavoidclup21aa1}
\bl{p_{err}^{(2)}} & - &  \mbox{probability of error after the second iteration}\nonumber \\
\bl{\hat{s}^{(2)}} & = &  \mE((\x^{(1)})^T\x^{(2,s)}) - \mbox{objective value after the second iteration} \nonumber \\
\bl{\hat{d}_2^{(2)}} & = & \mE\|\x^{(2,s)}\|_2^2 - \mbox{squared norm after the second iteration} \nonumber \\
\bl{\hat{d}_1^{(2)}} & = & \mE\x_{sol}^T\x^{(2,s)} - \mbox{inner product with $\x_{sol}$ after the second iteration}.
\end{eqnarray}
with the last three quantities being not only the expected but also the concentrating values as well.
Before proceeding with the numerical results for the second iteration we recall the output of the first iteration
\begin{eqnarray}\label{eq:numressec1}
  \phi^{(1)}  =  \{\bl{p_{err}^{(1)},\hat{s}^{(1)},\hat{d}_2^{(1)},\hat{d}_1^{(1)}},\prp{\hat{\nu}^{(1)},\hat{\gamma}^{(1)},\hat{c}_{1,z}^{(1)}}\}=  \{0.0072,-0,0.7574,0.8369,0,1.2233,0.835\}.
\end{eqnarray}
The theoretical values for the output parameters after the second iteration (i.e. for the parameters from (\ref{eq:sumavoidclup21aa}) that are obtained through the discussion presented above for SNR, $1/\sigma^2=13$[db], $\alpha=0.8$, and $r_{sc}=1.3$) are included in Table \ref{tab:tabsecavoidnum1}.
\begin{table}[h]
\caption{\textbf{Theoretical} values for various parameters at the output of the second iteration}\vspace{.1in}
\hspace{-0in}\centering
\small{
\begin{tabular}{||c||c|c|c||c||c|c|c|c|c||}\hline\hline
$1/\sigma^2 $[db]  & $\hat{\nu}^{(2)}$ & $\hat{\nu}_2^{(2)}$ & $\hat{\gamma}^{(2)}$ & $\hat{p}^{(1)}$ & $-\hat{s}^{(2)}$ &   $\xi_{RD}^{(2)} $ & $p_{err}^{(2)} $  & $\|\x^{(2,s)}\|_2^2$ &
$(\x_{sol})^T\x^{(2,s)}$ \\ \hline\hline
$13  $ & $\mathbf{2.7496}  $ & $\mathbf{-0.7821}  $ & $\mathbf{2.1388}  $ & $\mathbf{0.769}  $ & $\mathbf{0.9494}  $ & $\mathbf{0.1594}  $ & $\mathbf{0.00092}  $ & $\mathbf{0.9370}  $ & $\mathbf{0.9600}  $ \\ \hline\hline
\end{tabular}}
\label{tab:tabsecavoidnum1}
\end{table}
One can also characterize the remaining auxiliary parameters from $\phi^{(2)}$, i.e. $\prp{\{\hat{q}^{(1)},\hat{c}_{2,z}^{(2)},\hat{s}_{2}^{(2)},\hat{s}_{3}^{(2)}}\}$ relying on the equality constraints in (\ref{eq:sumsecavoidclup17a5}). Table \ref{tab:tabsecavoidnum2} shows the results for these parameters that can be obtained through both, the equality constraints in (\ref{eq:sumsecavoidclup17a5}) and (\ref{eq:sumsecavoidclup17a8}).
\begin{table}[h]
\caption{\textbf{Theoretical} values for $\prp{\{\hat{q}^{(1)},\hat{c}_{2,z}^{(2)},\hat{s}_{2}^{(2)},\hat{s}_{3}^{(2)}}\}$ obtained utilizing (\ref{eq:sumsecavoidclup17a5}) (\textbf{bold}) as well as (\ref{eq:sumsecavoidclup17a8}) (\prp{\textbf{purple}}) }\vspace{.1in}
\hspace{-0in}\centering
\small{
\begin{tabular}{||c||c|c|c|c||}\hline\hline
$1/\sigma^2 $[db]  & $\hat{s}_2^{(2)}=\hat{d}_1^{(1)}+\hat{s}^{(2)}\sqrt{\hat{d}_2^{(1)}}$ & $\hat{s}_3^{(2)}=1-\hat{d}_1^{(2)}$ & $\hat{c}_{2,z}^{(2)}=\hat{d}_2^{(2)}-2\hat{d}_1^{(2)}+1$ & $\hat{q}^{(1)}=\frac{\hat{s}_3-\hat{s}_2+\sigma^2}{\sqrt{\hat{c}_{2,z}+\sigma^2}\sqrt{\hat{c}_{1,z}+\sigma^2}}$ \\ \hline\hline
$13  $ & $\mathbf{0.01065}  $/\prp{$\mathbf{0.01065}  $} & $\mathbf{0.0400}  $/\prp{$\mathbf{0.0400}  $}&
$\mathbf{0.0170}$/\prp{$\mathbf{0.0170}$} & $\mathbf{0.83885}$/\prp{$\mathbf{0.83885}  $}    \\ \hline\hline
\end{tabular}}
\label{tab:tabsecavoidnum2}
\end{table}

\section{Summary of the CLuP-plt's $(k+1)$-th iteration analysis}
\label{sec:clupkit}

The heart of the analysis mechanism is the move from the first to the second iteration. Such a move is conceptually then identical to the move from any $k$-th to $(k+1)$-th iteration. However, there are still a few technical differences that require a special attention. These differences are of course the main reason why we separately discuss a generic move from $k$-th to $(k+1)$-th iteration for any $k>1$. On the other hand, we have already faced a similar situation in \cite{Stojnicclupcmpl19} and all the results obtained there in this regard can be reutilized. We
start by recalling that  CLuP's $(k+1)$-th iteration is basically the following optimization problem
\begin{eqnarray}
\x^{(k+1)}=\frac{\x^{(k+1,s)}}{\|\x^{(k+1,s)}\|_2} \quad \mbox{with}\quad \x^{(k+1,s)}=\mbox{arg}\min_{\x} & & -(\x^{(k)})^T\x  \nonumber \\
\mbox{subject to} & & \|\y-A\x\|\leq r\nonumber \\
&& \x\in \left [-\frac{1}{\sqrt{n}},\frac{1}{\sqrt{n}}\right ]^n. \label{eq:kit0}
\end{eqnarray}
This is of course structurally identical to (85) in \cite{Stojnicclupcmpl19}. One can then again utilize the Random Duality Theory and repeat all the steps between (85) and (108) in \cite{Stojnicclupcmpl19} to arrive at the following
for the optimizing $\z_i$ and $\x_i^{(k+1,s)}$
\begin{eqnarray}
\z_i^{(k+1)} & = & \frac{1}{\sqrt{n}}\min \left (\max\left (0,-\left (\frac{\h_i^{(k,p)}+\sum_{j=1}^{k}\tilde{\nu}_j\x^{(j,s)}_i+\nu_2}{2\gamma}\right )\right ),2\right )\nonumber \\
\x_i^{(k+1,s)} & = & \frac{1}{\sqrt{n}}-\z_i^{(2)}=\frac{1}{\sqrt{n}}\left (1-\min \left (\max\left (0,-\left (\frac{\h_i^{(k,p)}+\sum_{j=1}^{k}\tilde{\nu}_j\x^{(j,s)}_i+\nu_2}{2\gamma}\right )\right ),2\right )\right ),\label{eq:kit24}
\end{eqnarray}
where $\x^{(j,s)}_i,1\leq j\leq k$ are obtained after the $k$-th iteration as the optimizing variables after each of the first $k$ iterations. One can also set as in \cite{Stojnicclupcmpl19}'s (110)
\begin{equation}
I_{1}^{(k+1)}(\gamma,\nu,\nu_2,\hat{\nu}^{(1)})  =  \mE((\h_i^{(k,p)}+\sum_{j=1}^{k}\tilde{\nu}_j\x^{(j,s)}_i+\nu_2)\z_i^{(k+1)}+\gamma\left (\z_i^{(k+1)}\right )^2),
\label{eq:kit26}
\end{equation}
where for $\gamma<0$ the term under the expectation is assumed zero if negative. Moreover, one can also set as in \cite{Stojnicclupcmpl19}'s (111)
\begin{eqnarray}
\xi_{RD}^{(k+1)}(\alpha,\sigma;P^{(k+1)},Q^{(k+1)},c_{2,z},s_{2,j},s_3,\gamma,\tilde{\nu}_j,\nu_2) & = & \sqrt{\alpha}\sqrt{c_{2,z}+\sigma^2} f_{sph}^{(k+1)} +I_{1}^{(k+1)}(\gamma,\nu,\nu_2,\hat{\nu}^{(1)}) \nonumber \\
& &-\sum_{j=1}^{k}\tilde{\nu}_j s_{2,j}-\nu_2 s_3-\gamma c_{2,z}, \label{eq:kit27}
\end{eqnarray}
where $P^{(k+1)}$ and $Q^{(k+1)}$ are as in \cite{Stojnicclupcmpl19}'s (90) and $f_{sph}^{(k+1)}$ is as in \cite{Stojnicclupcmpl19}'s (101). We also note from \cite{Stojnicclupcmpl19}'s (112)-(114) that the key output parameters after the $k$-th iteration are
\begin{equation}
\x_i^{(j,s)},\z_i^{(j)},\lambda^{(j-1)}, 1\leq j\leq k,\label{eq:sumkit1}
\end{equation}
and
\begin{equation}\label{eq:sumkit2}
  \phi^{(k)}=\{\bl{p_{err}^{(k)},\hat{s}^{(k)},\hat{d}_2^{(k)},\hat{d}_1^{(k)}},\prp{\hat{\nu}^{(k)},\hat{\nu}_2^{(k)},\hat{\gamma}^{(k)},\hat{P}^{(k)},\hat{Q}^{(k)},\hat{c}_{2,z}^{(k)},\hat{s}_{2}^{(k)},\hat{s}_{3}^{(k)}}\},
\end{equation}
where $\hat{\nu}^{(k)}$ and $\hat{s}_{2}^{(k)}$ are the $(k-1)$-dimensional optimal $\tilde{\nu}$ and $\s_2$ vectors at the $k$-th iteration (i.e. for $k>1$, $\hat{\nu}^{(k)}=[\tilde{\nu}_1,\tilde{\nu}_2,\dots,\tilde{\nu}_{k-1}]$ for optimal $\tilde{\nu}_j$ and $\hat{s}_{2}^{(k)}=[s_{2,1},s_{2,2},\dots,s_{2,k-1}]$ for optimal ${s}_{2,j}$). We also recall from \cite{Stojnicclupcmpl19} that four probably most important output parameters after the $k$-th iteration are
\begin{eqnarray}\label{eq:sumkit3}
\bl{p_{err}^{(k)}} & - &  \mbox{probability of error after the $k$-th iteration}\nonumber \\
\bl{\hat{s}^{(k)}} & = &  \mE((\x^{(k-1)})^T\x^{(k,s)}) - \mbox{objective value after the $k$-th iteration} \nonumber \\
\bl{\hat{d}_2^{(k)}} & = & \mE\|\x^{(k,s)}\|_2^2 - \mbox{squared norm after the $k$-th iteration} \nonumber \\
\bl{\hat{d}_1^{(k)}} & = & \mE\x_{sol}^T\x^{(k,s)} - \mbox{inner product with $\x_{sol}$ after the $k$-th iteration}.
\end{eqnarray}
The above is what one essentially has avaialable before approaching handling the $(k+1)$-th iteration.

\vspace{.1in}
\noindent \xmyboxc{\bl{\emph{\textbf{Key part -- \dgr{Handling the $(k+1)$-th iteration}} }}}

Analogously to \cite{Stojnicclupcmpl19}'s (117) we have
\begin{eqnarray}
\phi_b^{(k+1)}=\mbox{arg} \min_{s,d_1^{(k+1)},d_2^{(k+1)},s_{2,j}} & & s\nonumber \\
\mbox{subject to} & & \max_{P^{(k+1)}}\min_{0\leq c_{2,z}\leq 4}\max_{\gamma,\nu,\nu_2}\xi_{RD}^{(k+1)}(\alpha,\sigma;P^{(k+1)},Q^{(k+1)},c_{2,z},s_{2,j},s_3,\gamma,\tilde{\nu}_j,\nu_2)=r\nonumber \\
& & s_{2,k}=\hat{d}_1^{(k)}+s\sqrt{\hat{d}_2^{(k)}}\nonumber \\
& & s_3=1-d_1^{(k+1)}\nonumber \\
& & c_{2,z}=d_2^{(k+1)}-2d_1^{(k+1)}+1\nonumber \\
& & Q_{k+1,j}^{(k+1)} =\frac{s_3-s_{2,j}+\sigma^2}{\sqrt{c_{2,z}+\sigma^2}\sqrt{\hat{c}_{2,z}^{(j)}+\sigma^2}},\label{eq:sumkit6}
\end{eqnarray}
with
\begin{equation}\label{eq:sumkit7}
\phi_b^{(k+1)}=\{\hat{P}^{(k+1)},\hat{Q}^{(k+1)},\hat{\nu}^{(k+1)},\hat{\nu}_2^{(k+1)},\hat{\gamma}^{(k+1)},\hat{s}^{(k+1)},\hat{d}_2^{(k+1)},\hat{d}_1^{(k+1)},\hat{s}_2^{(k+1)}\}.
\end{equation}
Moreover, analogously to \cite{Stojnicclupcmpl19}'s (118)-(119) we also have
\begin{equation}\label{eq:sumkit8}
  p_{err}^{(k+1)}=1-\mE((\mbox{sign}(\x^{(k+1,s)})+1)/2).
\end{equation}
Finally, in addition to the solution, $\x^{(k+1,s)}$, we also have the following as the full set of all \bl{critical} plus \prp{auxiliary} parameters that appear at the output of the $(k+1)$-th iteration:
\begin{equation}\label{eq:sumkit10}
  \phi^{(k+1)}=\{\bl{p_{err}^{(k+1)},\hat{s}^{(k+1)},\hat{d}_2^{(k+1)},\hat{d}_1^{(k+1)}},\prp{\hat{\nu}^{(k+1)},\hat{\nu}_2^{(k+1)},\hat{\gamma}^{(k+1)},\hat{P}^{(k+1)},\hat{Q}^{(k+1)},\hat{c}_{2,z}^{(k+1)},\hat{s}_{2}^{(k+1)},\hat{s}_{3}^{(k+1)}}\},
\end{equation}
where analogously to (\ref{eq:sumkit3})
\begin{eqnarray}\label{eq:sumki11}
\bl{p_{err}^{(k+1)}} & - &  \mbox{probability of error after the $(k+1)$-th iteration}\nonumber \\
\bl{\hat{s}^{(k+1)}} & = &  \mE((\x^{(k)})^T\x^{(k+1,s)}) - \mbox{objective value after the $(k+1)$-th iteration} \nonumber \\
\bl{\hat{d}_2^{(k+1)}} & = & \mE\|\x^{(k+1,s)}\|_2^2 - \mbox{squared norm after the $(k+1)$-th iteration} \nonumber \\
\bl{\hat{d}_1^{(k+1)}} & = & \mE\x_{sol}^T\x^{(k+1,s)} - \mbox{inner product with $\x_{sol}$ after the $(k+1)$-th iteration}.
\end{eqnarray}

\vspace{.1in}
\noindent \xmyboxc{\bl{\emph{\textbf{Numerical results -- third iteration ($k=2$)} }}}

As we have mentioned in \cite{Stojnicclupcmpl19}, the above discussion is in principle enough to compute all the critical parameters and basically fully characterize the CLuP-plt's performance. We will in a separate paper present a systematic way to compute all these parameters. In \cite{Stojnicclupcmpl19} we gave a quick estimate whose an analogue we could obtain here as well. However, as it will turn out later on, in the case of interest that we chose to highlight in this paper ($\alpha=0.8$, $1/\sigma^2=13$[db], $r_{sc}=1.3$), that is not necessary. Instead, we focus on what happens when $k=2$, i.e. in the third iteration.

First, we recall from (\ref{eq:numressec1}) that for the set of key output parameters after the first iteration we obtained the following
\begin{eqnarray}\label{eq:numresthird1}
  \phi^{(1)}  =  \{\bl{p_{err}^{(1)},\hat{s}^{(1)},\hat{d}_2^{(1)},\hat{d}_1^{(1)}},\prp{\hat{\nu}^{(1)},\hat{\gamma}^{(1)},\hat{c}_{1,z}^{(1)}}\}=  \{\mathbf{0.0072,-0,0.7574,0.8369,0,1.2233,0.835}\}.
\end{eqnarray}
This set of parameters is then utilized to obtain in Tables \ref{tab:tabsecavoidnum1} and \ref{tab:tabsecavoidnum2} a similar set of parameters after the second iteration (basically the one from (\ref{eq:sumavoidclup21aa}))
\begin{eqnarray}\label{eq:numresthird2}
  \phi^{(2)} & = & \hspace{-.02in}\{\bl{p_{err}^{(2)},\hat{s}^{(2)},\hat{d}_2^{(2)},\hat{d}_1^{(2)}},\prp{\hat{\nu}^{(2)},\hat{\nu}_2^{(2)},\hat{\gamma}^{(2)},\hat{p}^{(1)},\hat{q}^{(1)},\hat{c}_{2,z}^{(1)},\hat{s}_{2}^{(2)},\hat{s}_{3}^{(2)}}\}\nonumber \\
  & = &\hspace{-.02in} \{\mathbf{0.00092,-0.9494,0.937,0.96,2.7498,-0.7821,2.1388,0.769,0.8389,0.017,0.01065,0.04}\}.\nonumber \\
\end{eqnarray}
The $\phi^{(2)}$ set can then be utilized according to the above mechanism to obtain $\phi^{(3)}$ set from (\ref{eq:sumkit10}). In Tables \ref{tab:tabnumsecthird1} and \ref{tab:tabnumsecthird2} we show some of these results.
\begin{table}[h]
\caption{\textbf{Theoretical} values for various parameters at the output of the third iteration}\vspace{.1in}
\hspace{-0in}\centering
\small{
\begin{tabular}{||c||c|c|c||c||c|c|c|c||}\hline\hline
$1/\sigma^2 $[db]  & $\hat{\nu}^{(3)}$ & $\hat{\nu}_2^{(3)}$ & $\hat{\gamma}^{(3)}$ & $-\hat{s}^{(3)}$ &   $\xi_{RD}^{(3)} $ & $p_{err}^{(3)} $  & $\|\x^{(3,s)}\|_2^2$ &
$(\x_{sol})^T\x^{(3,s)}$ \\ \hline\hline
$13  $ & $[\mathbf{-0.0930,1.4480} ]$ & $\mathbf{-0.4}  $ & $\mathbf{2.0198}  $   & $\mathbf{0.9705}  $ & $\mathbf{0.1594}  $ & $\mathbf{0.00022}  $ & $\mathbf{0.9440}  $ & $\mathbf{0.9660}  $ \\ \hline\hline
\end{tabular}}
\label{tab:tabnumsecthird1}
\end{table}
In particular, we show the results for $\{\bl{p_{err}^{(3)},\hat{s}^{(3)},\hat{d}_2^{(3)},\hat{d}_1^{(3)}},\prp{\hat{\nu}^{(3)},\hat{\nu}_2^{(3)},\hat{\gamma}^{(3)}}\}$ in Table \ref{tab:tabnumsecthird1}. In Table \ref{tab:tabnumsecthird2}, we show the results for $\{\prp{\hat{c}_{2,z}^{(3)},\hat{s}_{2}^{(3)},\hat{s}_{3}^{(3)}}\}$. The second component of $\hat{s}_2^{(3)}$ based on (\ref{eq:sumkit6}) is $\hat{d}_1^{(2)}+\hat{s}^{(3)}\sqrt{\hat{d}_2^{(2)}}$.
\begin{table}[h]
\caption{\textbf{Theoretical} values for $\prp{\{\hat{c}_{2,z}^{(3)},\hat{s}_{2}^{(3)},\hat{s}_{3}^{(3)}}\}$  }\vspace{.1in}
\hspace{-0in}\centering
\small{
\begin{tabular}{||c||c|c|c||}\hline\hline
$1/\sigma^2 $[db]  & $\hat{s}_2^{(3)}$ & $\hat{s}_3^{(3)}=1-\hat{d}_1^{(3)}$ & $\hat{c}_{2,z}^{(3)}=\hat{d}_2^{(3)}-2\hat{d}_1^{(3)}+1$ \\ \hline\hline
$13  $ & $[\mathbf{0.0120,0.0206}]  $ & $\mathbf{0.0340} $&
$\mathbf{0.0120}$     \\ \hline\hline
\end{tabular}}
\label{tab:tabnumsecthird2}
\end{table}
The results for optimized matrices $P$ and $Q$ aer shown separately
\begin{equation}
  \hat{P}^{(3)}=\begin{bmatrix}
             1 & 0.769 & 0.771 \\
           0.769 &  1 & 0.984 \\
           0.771 & 0.984 &  1
                     \end{bmatrix}
\end{equation}
\begin{equation}
  \hat{Q}^{(3)}=\begin{bmatrix}
             1 & 0.839  & 0.792  \\
           0.839 &  1   & 0.984 \\
           0.792 & 0.984 &   1
           \end{bmatrix}.
\end{equation}
We also recall from (\ref{eq:sumkit6})
\begin{eqnarray}
\hat{Q}_{3,1}^{(3)} & = & \frac{\hat{s}_3^{(3)}-\hat{s}_{2,1}^{(3)}+\sigma^2}{\sqrt{\hat{c}_{2,z}^{(3)}+\sigma^2}\sqrt{\hat{c}_{2,z}^{(1)}+\sigma^2}}\nonumber \\
\hat{Q}_{3,2}^{(3)} & = & \frac{\hat{s}_3^{(3)}-\hat{s}_{2,2}^{(3)}+\sigma^2}{\sqrt{\hat{c}_{2,z}^{(3)}+\sigma^2}\sqrt{\hat{c}_{2,z}^{(2)}+\sigma^2}},\label{eq:numresthird3}
\end{eqnarray}
where $\hat{s}_{2,2}^{(3)}$ is the second component of vector $\hat{s}_{2}^{(3)}$ and for cosmetic reasons $\hat{c}_{2,z}^{(1)}=\hat{c}_{`,z}^{(1)}$. From (\ref{eq:sumsecavoidclup17a5}) and \cite{Stojnicclupcmpl19}'s (90) we also recall that
\begin{eqnarray}
\hat{Q}_{2,1}^{(3)} & = & q^{(1)} =\frac{\hat{s}_3^{(2)}-\hat{s}_{2,1}^{(2)}+\sigma^2}{\sqrt{\hat{c}_{2,z}^{(2)}+\sigma^2}\sqrt{\hat{c}_{2,z}^{(1)}+\sigma^2}},\label{eq:numresthird3}
\end{eqnarray}
where $\hat{s}_{2,1}^{(2)}$ is the first (and only) component of vector $\hat{s}_{2}^{(2)}$.

In Table \ref{tab:tabnumsecthird3}, we show the estimated values for $p_{err}^{(k)}$, $\hat{s}^{(k)}$, $\hat{d}_2^{(k)}$, and $\hat{d}_1^{(k)}$ and how they progress through the iterations
\begin{table}[h]
\caption{Change in $p_{err}^{(k)}$, $\hat{s}^{(k)}$, $\|\x^{(k,s)}\|_2^2$, and $(\x_{sol})^T\x^{(k,s)}$ as $k$ grows; $\alpha=0.8$; $r_{sc}=1.3$; CLuP-plt}
\vspace{.1in}
\hspace{-0in}\centering
\footnotesize{
\begin{tabular}{||c||c|c|c|c||}\hline\hline
$i$  & $p_{err}^{(i)}$ & $-\hat{s}^{(k)}$ & $\hat{d}_2^{(k)}=\|\x^{(i)}\|_2^2$ & $\hat{d}_1^{(k)}=(\x_{sol})^T\x^{(i)}$ \\ \hline\hline
$1$  & $\mathbf{0.00720  }$ & $\mathbf{0.0000  }$ & $\mathbf{0.7574  }$ & $\mathbf{0.8369  }$ \\ \hline
$2$  & $\mathbf{0.00092  }$ & $\mathbf{0.9494  }$ & $\mathbf{0.9370  }$ & $\mathbf{0.9600  }$ \\ \hline
$3$  & $\mathbf{0.00022  }$ & $\mathbf{0.9705  }$ & $\mathbf{0.9440  }$ & $\mathbf{0.9660  }$ \\ \hline\hline
limit & $\mathbf{0.00016}  $ & $\mathbf{0.9721}  $ & $\mathbf{0.9451}  $ & $\mathbf{0.9668}  $  \\ \hline\hline
\end{tabular}}
\label{tab:tabnumsecthird3}
\end{table}
We also add that the discussion from \cite{Stojnicclupcmpl19} regarding $f_{sph}^{(3)}$ applies here as well with the values even closer to one. Namely, from the above considerations and what we presented in \cite{Stojnicclupcmpl19}, we now have $f_{sph}^{(3)}\geq 0.9994$. One can now see what we hinted at a long while ago. Already after the \bl{\textbf{third iteration}} CLuP-plt approaches the limiting CLuP performance (such a performance does not have an \emph{a priori} restriction on the number of iterations). In fact, to be completely fair towards CLuP, the overall CluP-plt variant discussed above needs three iterations but the CLuP mechanism itself needs only \bl{\textbf{two iterations}} (the first iteration technically speaking is not really a part of the CLuP's inherent structure; instead it belongs to the polytope-relaxation heuristic).

\section{Simulations}
\label{sec:cluppltsim}

We in this section complement the above theoretical findings with a set of results obtained through numerical simulations. As mentioned above, we considered a standard scenario, $\alpha=0.8$, $1/\sigma^2=13$[db], $r_{sc}=1.3$ and run the simulations with $n=800$. The results are shown in Table \ref{tab:tabsim1}.
\begin{table}[h]
\caption{CLuP-\red{plt} -- change in $p_{err}^{(k)}$, $\hat{s}^{(k)}$, $\|\x^{(k,s)}\|_2^2$, and $(\x_{sol})^T\x^{(k,s)}$ as $k$ grows; $\alpha=0.8$; $r_{sc}=1.3$; \bl{\textbf{Simulated}} ($n=800$)/\textbf{Theory--computed} ($n\rightarrow\infty$)}\vspace{.1in}
\hspace{-0in}\centering
\footnotesize{
\begin{tabular}{||c||c|c|c|c||}\hline\hline
$k$  & $p_{err}^{(k)}$ & $-\hat{s}^{(k)}$ & $\hat{d}_2^{(k)}=\|\x^{(k,s)}\|_2^2$ & $\hat{d}_1^{(k)}=(\x_{sol})^T\x^{(k,s)}$ \\ \hline\hline
$1$  & \red{$\mathbf{0.00776  }$}/$\mathbf{0.00720  }$ &  \red{$\mathbf{0.0000  }$}/$\mathbf{0.0000  }$ &  \red{$\mathbf{0.7572  }$}/$\mathbf{0.7574  }$ & \red{$\mathbf{0.8365  }$}/$\mathbf{0.8369  }$ \\ \hline
$2$  & \bl{$\mathbf{0.00078  }$}/$\mathbf{0.00092  }$ & \bl{$\mathbf{0.9495  }$}/$\mathbf{0.9494  }$ & \bl{$\mathbf{0.9390  }$}/$\mathbf{0.9370  }$ & \bl{$\mathbf{0.9616  }$}/$\mathbf{0.9600  }$ \\ \hline
$3$  & \bl{$\mathbf{0.00270  }$}/$\mathbf{0.00022  }$ & \bl{$\mathbf{0.9710  }$}/$\mathbf{0.9705  }$ & \bl{$\mathbf{0.9444  }$}/$\mathbf{0.9440  }$ & \bl{$\mathbf{0.9661  }$}/$\mathbf{0.9660  }$ \\ \hline\hline
limit & $\mathbf{0.00016}  $ & $\mathbf{0.9721}  $ & $\mathbf{0.9451}  $ & $\mathbf{0.9668}  $  \\ \hline\hline
\end{tabular}}
\label{tab:tabsim1}
\end{table}
We observe a very good agreement between the theoretical predictions and the results obtained through numerical simulations. We also obtained the following simulated results for matrices $P$ and $Q$
\begin{equation}
  P^{(3)}=\begin{bmatrix}
    \bl{\mathbf{1}} &   \bl{\mathbf{0.7653}}  &  \bl{\mathbf{0.7660}}  \\
    \bl{\mathbf{0.7653}} &   \bl{\mathbf{1}}  &  \bl{\mathbf{0.9886}}  \\
    \bl{\mathbf{0.7660}} &   \bl{\mathbf{0.9886}}  &  \bl{\mathbf{1}}
    \end{bmatrix} \qquad \hat{P}^{(3)}=\begin{bmatrix}
            \mathbf{1} & \mathbf{0.769} & \mathbf{0.771} \\
           \mathbf{0.769} &  \mathbf{1} & \mathbf{0.984} \\
           \mathbf{0.771} & \mathbf{0.984} &  \mathbf{1}
                     \end{bmatrix}
\end{equation}
and
\begin{equation}
  Q^{(3)}=\begin{bmatrix}
    \bl{\mathbf{1}} &   \bl{\mathbf{0.8300}}  &  \bl{\mathbf{0.7915}}  \\
    \bl{\mathbf{0.8300}} &   \bl{\mathbf{1}}  &  \bl{\mathbf{0.9893}} \\
    \bl{\mathbf{0.7915}} &   \bl{\mathbf{0.9893}}  &  \bl{\mathbf{1}}
                                \end{bmatrix} \qquad \hat{Q}^{(3)}=\begin{bmatrix}
            \mathbf{ 1} & \mathbf{0.839}  & \mathbf{0.792}  \\
           \mathbf{0.839} &  \mathbf{1}   & \mathbf{0.984} \\
           \mathbf{0.792} & \mathbf{0.984} &   \mathbf{1}
           \end{bmatrix}.
\end{equation}
We again observe a very solid agreement between the theoretical and simulated results.

For a comparison, we in Table \ref{tab:tabsim2} show how the performance of the CLuP algorithm itself progresses through the iterations for the same parameters as above.
\begin{table}[h]
\caption{CLuP -- change in $p_{err}^{(k)}$, $\hat{s}^{(k)}$, $\|\x^{(k,s)}\|_2^2$, and $(\x_{sol})^T\x^{(k,s)}$ as $k$ grows; $\alpha=0.8$; $r_{sc}=1.3$; $n=800$; \bl{\textbf{Simulated}}}\vspace{.1in}
\hspace{-0in}\centering
\footnotesize{
\begin{tabular}{||c||c|c|c|c||}\hline\hline
$k$  & $p_{err}^{(k)}$ & $\hat{s}^{(k)}$ & $\hat{d}_2^{(k)}=\|\x^{(k,s)}\|_2^2$ & $\hat{d}_1^{(k)}=(\x_{sol})^T\x^{(k,s)}$ \\ \hline\hline
$1$  & $\bl{\mathbf{0.04470  }}$ & $\bl{\mathbf{0.1294  }}$ & $\bl{\mathbf{0.7046  }}$ & $\bl{\mathbf{0.7658  }}$ \\ \hline
$2$  & $\bl{\mathbf{0.00790  }}$ & $\bl{\mathbf{0.9126  }}$ & $\bl{\mathbf{0.9048  }}$ & $\bl{\mathbf{0.9317  }}$ \\ \hline
$3$  & $\bl{\mathbf{0.00121  }}$ & $\bl{\mathbf{0.9618  }}$ & $\bl{\mathbf{0.9374  }}$ & $\bl{\mathbf{0.9603  }}$ \\ \hline
$4$  & $\bl{\mathbf{0.00033  }}$ & $\bl{\mathbf{0.9705  }}$ & $\bl{\mathbf{0.9438  }}$ & $\bl{\mathbf{0.9658  }}$ \\ \hline
$5$  & $\bl{\mathbf{0.00020  }}$ & $\bl{\mathbf{0.9719  }}$ & $\bl{\mathbf{0.9449  }}$ & $\bl{\mathbf{0.9668  }}$ \\ \hline\hline
limit & $\mathbf{0.00016}  $ & $\mathbf{0.9721}  $ & $\mathbf{0.9451}  $ & $\mathbf{0.9668}  $  \\ \hline\hline
\end{tabular}}
\label{tab:tabsim2}
\end{table}
We observe that it takes one iteration less for CLuP-plt to get to a better performance level than the original CLuP. One should here also keep in mind what we mentioned earlier. Namely, the first iteration in CLuP-plt is not really a part of the CLuP structure itself, whereas the first iteration in \cite{Stojnicclupcmpl19} is. That basically means that starting CLuP with a better initial $\x^{(0)}$ in this case saves two out of four iterations.

\subsection{Changing SNR}
\label{sec:cluppltsimchangesnr}

We also simulated two different SNR scenarios as well. We first decreased the SNR to $1/\sigma^2=12$[db] and then to $1/\sigma^2=11$[db] while havng all other parameters take the same values as in the above $1/\sigma^2=13$[db] case. The results are shown in Tables \ref{tab:tabsim3} and \ref{tab:tabsim4}.
\begin{table}[h]
\caption{CLuP-\red{plt} -- change in $p_{err}^{(k)}$, $\hat{s}^{(k)}$, $\|\x^{(k,s)}\|_2^2$, and $(\x_{sol})^T\x^{(k,s)}$ as $k$ grows; $1/\sigma^2=12$[db]; $\alpha=0.8$; $r_{sc}=1.3$; $n=800$; \bl{\textbf{Simulated}}}\vspace{.1in}
\hspace{-0in}\centering
\footnotesize{
\begin{tabular}{||c||c|c|c|c||}\hline\hline
$k$  & $p_{err}^{(k)}$ & $\hat{s}^{(k)}$ & $\hat{d}_2^{(k)}=\|\x^{(k,s)}\|_2^2$ & $\hat{d}_1^{(k)}=(\x_{sol})^T\x^{(k,s)}$ \\ \hline\hline
$1$  & $\red{\mathbf{0.01640  }}$ & $\red{\mathbf{0.0000  }}$ & $\red{\mathbf{0.7360  }}$ & $\red{\mathbf{0.8136  }}$ \\ \hline
$2$  & $\bl{\mathbf{0.00371  }}$ & $\bl{\mathbf{0.9389  }}$ & $\bl{\mathbf{0.9263  }}$ & $\bl{\mathbf{0.9497  }}$ \\ \hline
$3$  & $\bl{\mathbf{0.00146  }}$ & $\bl{\mathbf{0.9654  }}$ & $\bl{\mathbf{0.9347  }}$ & $\bl{\mathbf{0.9578  }}$ \\ \hline
$4$  & $\bl{\mathbf{0.00094  }}$ & $\bl{\mathbf{0.9673  }}$ & $\bl{\mathbf{0.9362  }}$ & $\bl{\mathbf{0.9594  }}$ \\ \hline
$5$  & $\bl{\mathbf{0.00085  }}$ & $\bl{\mathbf{0.9677  }}$ & $\bl{\mathbf{0.9366  }}$ & $\bl{\mathbf{0.9597  }}$ \\ \hline\hline
limit & $\mathbf{0.00072}  $ & $\mathbf{0.9695}  $ & $\mathbf{0.9400}  $ & $\mathbf{0.9622}  $  \\ \hline\hline
\end{tabular}}
\label{tab:tabsim3}
\end{table}
In particular, in Table \ref{tab:tabsim3} we show the results for $1/\sigma^2=12$[db]. Again after a fairly small number of iterations one approaches the limiting performance. Moreover, although one is now a bit closer to the line of corrections (and supposedly a tougher to handle SNR regime), the increase in the number of iterations is rather minimal. Instead of three iterations that were needed for $1/\sigma^2=13$[db] here five iterations suffice. The same discussion regarding counting the starting iterations that we emphasized above applies again. That means that in terms of CLuP's own iterations, instead of two for $1/\sigma^2=13$[db] one now needs four for $1/\sigma^2=12$[db].

In Table \ref{tab:tabsim3} we show the results for $1/\sigma^2=11$[db]. One is now really close to the line of corrections and a significant increase in the number of iterations might be expected. However, as results in the table show within 8 iterations one is reaching performance level literally identical to the optimal one. However, already after the fifth iteration one is very close to the optimum with the margin of error being on the fourth decimal (of course this iteration numbering accounts for the first iteration; if one is more fair towards CLuP then these 8 and 5 iterations should be replaced by 7 and 4 of CLuP's own iterations).
\begin{table}[h]
\caption{CLuP-\red{plt} -- change in $p_{err}^{(k)}$, $\hat{s}^{(k)}$, $\|\x^{(k,s)}\|_2^2$, and $(\x_{sol})^T\x^{(k,s)}$ as $k$ grows; $1/\sigma^2=11$[db]; $\alpha=0.8$; $r_{sc}=1.3$; $n=800$; \bl{\textbf{Simulated}}}\vspace{.1in}
\hspace{-0in}\centering
\footnotesize{
\begin{tabular}{||c||c|c|c|c||}\hline\hline
$k$  & $p_{err}^{(k)}$ & $\hat{s}^{(k)}$ & $\hat{d}_2^{(k)}=\|\x^{(k,s)}\|_2^2$ & $\hat{d}_1^{(k)}=(\x_{sol})^T\x^{(k,s)}$ \\ \hline\hline
$1$  & $\red{\mathbf{0.02720  }}$ & $\red{\mathbf{0.0000  }}$ & $\red{\mathbf{0.7186  }}$ & $\red{\mathbf{0.7912  }}$ \\ \hline
$2$  & $\bl{\mathbf{0.01117  }}$ & $\bl{\mathbf{0.9304  }}$ & $\bl{\mathbf{0.9172  }}$ & $\bl{\mathbf{0.9359  }}$ \\ \hline
$3$  & $\bl{\mathbf{0.00585  }}$ & $\bl{\mathbf{0.9619  }}$ & $\bl{\mathbf{0.9298  }}$ & $\bl{\mathbf{0.9498  }}$ \\ \hline
$4$  & $\bl{\mathbf{0.00384  }}$ & $\bl{\mathbf{0.9654  }}$ & $\bl{\mathbf{0.9332  }}$ & $\bl{\mathbf{0.9540  }}$ \\ \hline
$5$  & $\bl{\mathbf{0.00304  }}$ & $\bl{\mathbf{0.9663  }}$ & $\bl{\mathbf{0.9342  }}$ & $\bl{\mathbf{0.9554  }}$ \\ \hline
$6$  & $\bl{\mathbf{0.00281  }}$ & $\bl{\mathbf{0.9666  }}$ & $\bl{\mathbf{0.9346  }}$ & $\bl{\mathbf{0.9560  }}$ \\ \hline
$7$  & $\bl{\mathbf{0.00263  }}$ & $\bl{\mathbf{0.9668  }}$ & $\bl{\mathbf{0.9347  }}$ & $\bl{\mathbf{0.9562  }}$ \\ \hline
$8$  & $\bl{\mathbf{0.00255  }}$ & $\bl{\mathbf{0.9668  }}$ & $\bl{\mathbf{0.9348  }}$ & $\bl{\mathbf{0.9563  }}$ \\ \hline\hline
limit & $\mathbf{0.00249}  $ & $\mathbf{0.9670}  $ & $\mathbf{0.9350}  $ & $\mathbf{0.9565}  $  \\ \hline\hline
\end{tabular}}
\label{tab:tabsim4}
\end{table}

\section{Conclusion}
\label{sec:conc}

In \cite{Stojnicclupint19} we introduced the so-called CLuP (Controlled Loosening-up) method as a way to solve the MIMO ML problem \textbf{exactly}. As observed already in \cite{Stojnicclupint19}, one of the CLuP's very best features is its very low running complexity. Basically, as an iterative algorithm, its complexity is mainly driven by the number of iterations that it needs to achieves required convergence precision. Along those lines, what was essentially observed in \cite{Stojnicclupint19} was the fact that the typical number of iterations needed to achieve an excellent performance is not only polynomial but actually rather a fixed number that does not depend on the problem dimension. One would typically assume that if a polynomial number of iterations suffices then that would have already been an unprecedent feature in algorithms that attack MIMO ML. The discovery that a fixed number of iterations works as well was beyond remarkable.

We then in the followup paper \cite{Stojnicclupcmpl19} looked at this property more carefully and designed a Random Duality Theory type of mechanism to precisely quantify not only the required number of iterations but rather the behavior of all important systems parameters as they move/change through the CLuP's iterations. Such an approach is of course the most complete type of performance analysis that one can hope for. As expected, the analysis confirmed all the observations from \cite{Stojnicclupint19}. As we wanted to maintain the introductory papers on this topic to be related to the simplest possible underlying structures, we in \cite{Stojnicclupint19,Stojnicclupcmpl19} considered only the most basic CLuP version. However, on numerous occasions we did emphasize that various more advanced structures can now easily be built and analyzed. In this paper, we provide a first step in those directions. Namely, the standard basic CLuP from \cite{Stojnicclupint19,Stojnicclupcmpl19} is here modified to its a variant where for the starting step of the algorithm instead of a random initialization one utilizes the well-known so-called polytope-relaxation heuristic. We provided again a detailed \textbf{\emph{per iteration level analysis}} similar to the one that we provided in \cite{Stojnicclupcmpl19} for the standard CLuP. It turns out that the new version of CLuP, to which we refer as CLuP-plt, is indeed faster and requires a smaller number of iterations.

We should here also emphasize as in \cite{Stojnicclupint19,Stojnicclupcmpl19}, that we again didn't utilize the most advanced concepts but rather a very simple upgrade. As earlier, we wanted to showcase the conceptual opportunity for upgrade rather than its a best realization (we will address those in separate papers that will deal a bit more with further engineering of the main concepts rather than with their fundamental structuring). Nonetheless, the fact that in the regimes of interest CLuP-plt needed between 3-5 iterations in total (2-4 if one excludes the initialization) continues to sound almost unbelievable. After the theoretical considerations, we proceeded further and presented a solid set of numerical simulations results. They turned out to be in an excellent agreement with the theoretical predictions.

As mentioned above, a large class of way more sophisticated CLuP versions we will discuss in separate papers. Moreover, we will also in parallel present how they behave when applied for solving problems from different scientific fields as well.

\begin{singlespace}
\bibliographystyle{plain}
\bibliography{cluppltRefs}
\end{singlespace}

\end{document}